\documentclass[journal]{IEEEtran}
\usepackage{graphicx}
\usepackage{epsfig}
\usepackage{amsmath}
\usepackage{amsmath,epsfig,bm,algorithm,algorithmic,float}
\usepackage{amstext,graphicx,amssymb,amsfonts,xcolor,enumerate}
\usepackage{mathrsfs,cancel}
\usepackage{graphicx}
\usepackage{epstopdf}
\usepackage{epsfig}
\usepackage{amsmath,bm,algorithm,algorithmic,float}
\usepackage{amssymb}
\usepackage[thmmarks,amsmath,amsthm,hyperref]{ntheorem}
\usepackage{multirow}
\usepackage{slashbox, color}
\newtheorem{theorem}{Theorem}

\ifCLASSINFOpdf
\else
\fi
\begin{document}
%
\title{Fast Greedy Approaches for Compressive Sensing of Large-Scale Signals}
%
%
%

\author{Sung-Hsien Hsieh$^{*,**}$,
        Chun-Shien Lu$^{**}$, and
        Soo-Chang Pei$^{*}$\\
        $^{*}$Graduate Inst. Comm. Eng., National Taiwan University, Taipei, Taiwan\\
        $^{**}$Institute of Information Science, Academia Sinica, Taipei, Taiwan
\thanks{Corresponding author: Chun-Shien Lu; e-mail: lcs@iis.sinica.edu.tw}
}

\maketitle

\begin{abstract}
Cost-efficient compressive sensing is challenging when facing large-scale data, {\em i.e.}, data with large sizes.
Conventional compressive sensing methods for large-scale data will suffer from low computational efficiency and massive memory storage.
In this paper, we revisit well-known solvers called greedy algorithms, including Orthogonal Matching Pursuit (OMP), Subspace Pursuit (SP), Orthogonal Matching Pursuit with Replacement (OMPR).
Generally, these approaches are conducted by iteratively executing two main steps:
1) support detection and 2) solving least square problem.

To reduce the cost of Step 1, it is not hard to employ the sensing matrix that can be implemented by operator-based strategy instead of matrix-based one and can be speeded by fast Fourier Transform (FFT).
Step 2, however, requires maintaining and calculating a pseudo-inverse of a sub-matrix, which is random and not structural, and, thus, operator-based matrix does not work.
To overcome this difficulty, instead of solving Step 2 by a closed-form solution, we propose a fast and cost-effective least square solver, which combines a Conjugate Gradient (CG) method with our proposed weighted least square problem to iteratively approximate the ground truth yielded by a greedy algorithm.
Extensive simulations and theoretical analysis validate that the proposed method is cost-efficient and is readily incorporated with the existing greedy algorithms to remarkably improve the performance for large-scale problems.

\end{abstract}


%
\IEEEpeerreviewmaketitle

\section{Introduction}\label{Sec: intro}
In this section, we first briefly introduce the background of compressive sensing (CS) in Sec. \ref{Sec: BG}.
Then, the existing CS recovery algorithms (for large-scale signals) are discussed in Sec. \ref{Sec: RW}.
Finally, the overview and contributions of our proposed method are described in Sec. \ref{Sec: Overview}, followed by the organization of the remainder of this paper.

\subsection{Background}\label{Sec: BG}
Compressive sensing (CS) \cite{Donoho2006}\cite{Baraniuk2007}\cite{Candes2008} for sparse signals in achieving simultaneous data acquisition and compression has been extensively studied in the literature.
CS is recognized to be composed of fast encoder and slow decoder.

Let $x \in \mathbb{R}^{N}$ denote a $K$-sparse $1$-D signal to be sensed, let $\Phi \in \mathbb{R}^{M \times N}$  ($M < N$) represent a sampling matrix, and let $y \in \mathbb{R}^{M}$ be the measurement vector.
At the encoder, a signal $x$ is simultaneously sensed and compressed via $\Phi$ to obtain a so-called measurement vector $y$ as:
\begin{equation}
y=\Phi x,
\label{Eq: Random Projection}
\end{equation}
which is usually called a procedure of ranfom projection.
The measurement rate, defined as $0<\frac{M}{N}<1$, indicates the compression ratio (without quantization) and is a major concern in many applications.
\cite{Donoho2006}\cite{Candes2008} show that it is a good choice to design $\Phi$ as a Gaussian random matrix to satisfy either mutual incoherence property (MIP) or restricted isometry property (RIP).
Moreover, sparsity is an inherent assumption made in compressed sensing to solve the underdetermined system in Eq. (\ref{Eq: Random Projection}) due to $M < N$.
Nevertheless, for real applications, natural signals are often not sparse in either the time or space domain but can be sparsely represented in a transform (\textit{e.g.}, discrete cosine transform (DCT) or wavelet) domain.
Namely, $x=\Psi s$, where $\Psi$ is a transform basis (or dictionary) and $s$ is a sparse representation with respect to $\Psi$.
So, Eq. (\ref{Eq: Random Projection}) is also rewritten as:
\begin{equation}
y=\Phi\Psi s=As,
\label{Eq: Random Projection-2}
\end{equation}
where $A=\Phi\Psi \in \mathbb{R}^{M\times N}$.
We say that $x$ is $K$-sparse if $s$ contains only $K$ non-zero entries (exactly $K$-sparse) or $K$ significant components (approximately $K$-sparse).

At the decoder,  the original signal $x$ can be perfectly recovered by an intuitive solution to CS recovery, called $\ell_{0}$-minimization, which is defined as:
\begin{equation}
	\min_{s} \| s \|_{0} \quad s.t. \quad \| y-As \|_{2} \leq \epsilon,
\label{Eq: L0}
\end{equation}
where $\epsilon$ is a tolerable error term.
Due to $M < N$, this system is underdetermined and there exists infinite solutions.
Thus, solving $\ell_{0}$-minimization problem requires combinatorial search and is NP-hard.

Alternative solutions to Eq. (\ref{Eq: L0}) usually are  based on two strategies: convex programming and greedy algorithms.
For convex programming, researchers \cite{Donoho2006}\cite{Candes2008} have shown that when $M \geq O(K \log \frac{N}{K})$ holds, solving $\ell_{0}$-minimization is equivalent to solving $\ell_{1}$-minimization, defined as:
\begin{equation}
	\min_{s} \| s \|_{1} \quad s.t. \quad \| y-As \|_{2} \leq \epsilon.
\label{Eq: L1}
\end{equation}
Typical $\ell_{1}$-minimization models include Basis Pursuit (BP) \cite{Chen2008} and Basis Pursuit De-Noising (BPDN) \cite{Needella2003} with the computational complexity of recovery being polynomial.
For greedy approaches, including Orthogonal Matching Pursuit (OMP) \cite{Tropp2007}, CoSaMP \cite{Needella2009}, and Subspace pursuit (SP) \cite{Dai2009}, utilize a greedy strategy for support detection first and then solve a least square problem to recover the original signal.
The main difference among these greedy algorithms is how the support detection step is condcuted.

Nevertheless, for large-scale signals (\textit{e.g.}, $N>2^{20}$), both $\ell_{1}$-minimization and greedy algorithms suffer from high computational complexity and massive memory usage.
Ideally, the memory costs at the encoder and decoder are expected to approximate $O(M)$ and $O(N)$, respectively, which are the minimum costs to store the original measurement vector $y$ and signal $x$.
If $A$ is required to be stored completely, however, it will cost $O(MN)$ bytes ({\em e.g.}, when $N=2^{20}$ and $M=2^{18}$, the sensing matrix needs several terabytes).
It often overwhelms the capability of existing hardware devices.

In view of the incoming big data era, such a troublesome problem needs to pay immediate attention.
In this paper, we say that $X \in \mathbb{R}^{N_1\times N_2\times\cdots\times N_D}$ is a kind of big data if ${N_1\times N_2\times\cdots\times N_D}$ is lare enough or more specifically its size approaches the storage limit of hardwares like PC, notebook, and so on.

\subsection{Related Work}\label{Sec: RW}
The existing methods that can deal with compressive sensing of large-scale signals are discussed in this section. As mentioned above, we focus on computational efficiency and memory usage.
Basically, our survey is conducted from the aspects of encoder and decoder in compressive sensing.
We mainly discuss block-based, tensor-based, and operator-based compressive sensing algorithms here.

\begin{table*}[!htbp]
\caption{Comparison between each algorithm.}
\centering
\begin{tabular}{|c|c|c|c|c|}
\hline
algorithms & sensing strategy & assumptions & algorithm type & storage\\
\hline
N-BOMP \cite{Caiafa2013} & tensor-based (2D) & block sparsity & greedy & $O(\sqrt{MN})$\\
\hline
\cite{Caiafa2015} & tensor-based (2D) & low multilinear-rank & closed-form & $O(\sqrt{MN})$\\
\hline
BCS \cite{Gan2007} & block-based (1D) & - & Landweber-based & $O(\frac{MN}{B})$ \\
\hline
BCS-SPL \cite{Mun2009} & block-based (1D)& - & Landweber-based & $O(\frac{MN}{B})$ \\
\hline
\cite{Milzarek2014} & conventional (1D)& - & Armijo-based & $O(MN)$\\
\hline
GPSR \cite{Figueiredo2007}, SpaRSA \cite{Wright2008} & conventional (1D) & - & IST-based & $O(MN)$ \\
\hline
\cite{Beck2009}\cite{Wen2010} & conventional (1D) & - & FPC\_AS-based & $O(MN)$ \\
\hline
\end{tabular}
\label{table: comparison_algorithm}
\end{table*}

\subsubsection{Strategies at Encoder}\label{Sec: St encoder}
From Eq. (\ref{Eq: Random Projection}), we can see that both the storage (for $\Psi$) and computation (for $\Psi x$) costs require $O(MN)$ bytes and $O(MN)$ operations, respectively.
When the signal length becomes large enough, storing $\Phi$ and computing $\Phi x$  become an obstacle.

In the literature, Gan \cite{Gan2007} and Mun and Fowler \cite{Mun2009} propose block-based compressive sensing techniques, wherein a large-scale signal is separated into several small block signals, which are individually sensed via the same but smaller sensing matrix.
The structure of block sensing reduces both storage and computation costs to $O(\frac{MN}{B})$, where $B$ is the number of blocks.
Although block-based compressive sensing can deal with small blocks quickly and easily, it actually cannot work for
the scenario of medical imaging in that an image  generated from the fast Fourier Transform (FFT) coefficients of an entire sectional view \cite{Mun13} violates the structure of block-based sensing.

Shi {\em et al.} \cite{shi2013} and Caiafa and Cichocki \cite{Caiafa2013} consider the problem of large-scale compressive sensing based on tensors.
In other words, the signal is directly sensed and reconstructed in the original (high) dimensional space instead of reshaping to $1$-D.
For example, a $2$-D image $X \in \mathbb{R}^{\sqrt{N}\times \sqrt{N}}$ is sensed via
\begin{equation}
Y=\Phi_{1} X \Phi_{2}^{T},
\label{Eq: 2D Random Projection}
\end{equation}
where $\Phi_{1}$ and $\Phi_{2} \in \mathbb{R}^{\sqrt{M}\times \sqrt{N}}$, and $Y \in \mathbb{R}^{\sqrt{M}\times \sqrt{M}}$.
This strategy is often called separable sensing \cite{Rivenson09a, Rivenson09b}.
In this case, both the storage and computation costs are reduced to $O(\sqrt{MN})$.
\cite{Caiafa2015} further presents a close-from solution for reconstruction from compressive sensing based on assuming the low-rank structure.
It should be noted that since tensor-based approaches, in fact, change the classical sensing structure ({\em i.e.}, $y=\Phi x$) of CS, the decoder no longer follows the conventional solvers like Eq. (\ref{Eq: L1}).
Specifically, the measurements in tensor-based approaches form a tensor but conventional solvers only accept one-dimensional measurement vector.

In addition to block-based and tensor-based approaches, operator-based approaches are to design $\Phi$ as a deterministic matrix or structurally random matrix, implemented by certain fast operators.
For example, Candes {\em et al.} \cite{Candes2006-F} propose the use of a randomly-partial Fourier matrix as $\Phi$.
In this case, we can implement $\Phi x$ by $\mathcal{D}\left( FFT(x)\right)$, where $FFT(\cdot)$ is the function of fast Fourier transform (FFT) and $\mathcal{D}\left(\cdot \right)$ denotes a downsampling operator that outputs an $M \times 1$ vector.
Thus, $\Phi$ is not necessarily stored in advance.
In addition, the computation cost also becomes $O( N \log N)$, which especially outperforms $O(MN)$ for large-scale signals because $M$ is positively proportional to $N$.
Do {\em et al.} \cite{Do2012} further propose a kind of random Gaussian-like matrices, called Structurally Random Matrix (SRM), which benefits from operator-based strategy and achieves reconstruction performance as good as random Gaussian matrix.
In sum, since operator-based approaches follow the original CS structure, the decoder is not necessary to be modified.

\subsubsection{Strategies at Decoder}\label{Sec: St Decoder}
For block-based approaches \cite{Gan2007}\cite{Mun2009}, each block can be individually recovered with low computation cost and memory usage but incurs blocky effects between boundaries of blocks.
In \cite{Mun2009}, Mun and Fowler propose a method, called BCS-SPL, which further removes blocky effects by Wiener filtering.
In addition to the incapabliity of sensing medical images like MRI, BCS-SPL is also not adaptive in that the measurement rates are fixed for different blocks by ignoring the potential differences in smoothing blocks that need less measurement rates and complex blocks that require more measurement rates.

For tensor-based compressive sensing, \cite{Caiafa2013} develops a new solver called N-way block OMP (N-BOMP).
Though N-BOMP is indeed faster than conventional CS solvers, its performance closely depends on the unique sparsity pattern, {\em i.e.}, block sparsity, of an image.
Specifically, block sparsity states that the importnat components of an images are clustered together in blocks.
This characteristic seems to only naturally appear in hyperspectral imaging.
In \cite{Sidiropoulos2012}, a multiway compressive sensing (MWCS) method for sparse and low-rank tensors is proposed.
MWCS achieves more efficient reconstruction, but its performance relies heavily on tensor rank estimation, which is NP-hard.
A generalized tensor compressive sensing (GTCS) method \cite{Li2013}, which combines $\ell_{1}$-minimization with high-order tensors, is beneficial for parallel computation.

For operator-based compressive sensing algorithms, since the conventional solvers, mentioned in the previous subsection,  still can be used, here we mainly review state-of-the-art convex optimization algorithms focusing on the large-scale problem, where only simple operations such as $A$ and $A^{T}$ conducted by operator are required.

Cevher {\em et al.} \cite{Cevher2014} point out that an optimization algorithm based on the first-order method such as gradient descent features nearly dimension-independent convergence rate and is theoretically robust to the approximations of their oracles.
Moreover, the first-order method such as NESTA \cite{Becker2011} often involves the transpose of sensing matrix, which is easily implemented by operator.
Both GPSR \cite{Figueiredo2007} and SpaRSA \cite{Wright2008} are closely related to iterative shrinkage/threshold (IST) methods and support the operator-based strategy.
In addition, \cite{Beck2009}\cite{Wen2010}\cite{Wen2012} have shown that algorithms based on solving fixed-point equation have fast convergence rate, which can be combined with operator-based strategy too.
For example, Milzarek{\em et al.} \cite{Milzarek2014} further propose a globalized semismooth Newton method, where partial DCT matrix is adopted as the sensing matrix for fast sensing.
But, it requires that signals are sparse in the time/spatial domain leading to limited applications.


\subsubsection{Brief Summary of Related Works}\label{Sec: Brief Summary}
Table \ref{table: comparison_algorithm} depicts the comparisons among the aforementioned algorithms, where storage  is estimated based on non-operator version.
If operator can be used, the storage of storing a sensing matrix is not required and, thus, is bounded by $O(N)$ for each row vector, which is only related to the minimum requirement for storing the reconstructed signal.
Since the characteristic of compressive sensing states that CS encoder spends lower memory and computation cost than CS decoder, when taking hardware implementation in real world into consideration, tensor-based methods are more complicated than others.
For example, the single-pixel camera designed in \cite{Marco2008} uses a DMD array as a row of $A$ to sense $x$.
By changing the pattern of DMD array $M$ times, the measurements are collected.
This structure, however, cannot support separable sensing that is commonly used in tensor-based methods.
Block-based CS methods do not intrinsically overcome large-scale problems and lack convincing theoretical proof about complexity, performance, and convergence analysis.
Operator-based CS methods maintain the original structures of CS encoder and decoder.
Thus, most of the existing fast algorithms for $\ell_{1}$-minimization  can be used only if all of matrix operations can be executed in an operator manner.
Furthermore, they have strong theoretical validation since $\ell_{1}$-minimization is a well-known model and has been developed for years.
In fact, the operator-based strategy can also be employed in tensor-based and block-based compressive sensing methods to partially reduce their computation cost and storage usage.

\subsection{Contributions and Overview of Our Method}\label{Sec: Overview}

Up to now, it is still unclear how greedy algorithms can deal with large-scale problems by utilizing operator-based strategy.
Although SparseLab releases the OMP code combined with operator, the program still cannot deal with large-scale signals.
This challenge is the objective of this paper and, to our knowledge, we are the first to explore this issue.
In fact, our idea can help all greedy algorithms to deal with large-scale signals.
We will discuss the problem in detail in Sec. \ref{sec:Problem Formulation}.


Generally, greedy algorithms are conducted by iteratively executing two main stages: (a) support detection and (b) solving least square problem with the known support.
To reduce the cost of support detection, we follow the common strategy of adopting operator-based, instead of matrix-based, design of a sensing matrix ({\em e.g.}, \cite{Do2012}).
Therefore, we no longer discuss this step in this paper as it is not the focus of our method.

For solving the least square problem, we propose a fast and cost-effective solver by combining a Conjugate Gradient (CG) method with a weighted least square model to iteratively approximate the ground truth.
In our method, the memory cost of solving least square problem is reduced to $O(N)$, and
the computation cost of CG method is approximately $O(N \log N)$ for finite floating point precision and $O(KN \log N)$ for exact precision.

In should be noted that although using CG to solve the least square problem is not new, our extended use of CG brings additional advantages.
For example, Blumensath {\em et al.} \cite{Blumensath2008}  proposed ``Gradient Pursuits (GP)'' in which the memory cost is dominated by $O(MN)$ to save $\Phi$ (see Table 1 in \cite{Blumensath2008}), which cannot be stored explicitly for large-scale problem.
Our method extends GP to reduce the memory cost by using SRM to avoid saving $\Phi$.
In addition, we reformulate a least square problem used in GP into a weighted one and show both models are equivalent.
More specifically, solving weighted least square problem only requires $\Phi$ and can benefit from fast computation of operator-based approaches (as in SRM).
Traditional least square problem, however, involves sub-matrices of $\Phi$ and cannot directly be conducted by fast operator.

On the other hand, ``iteratively reweighted least-square (IRLS)'' was proposed in \cite{Chartrand2008}.
Though both IRLS and our method involve weighting, they are totally different.
First, IRLS uses weighting to approximate $\ell_1$-norm solution instead of $\ell_2$-norm solution in original least square problem while our greedy method still solves $\ell_2$-norm solution in the weighted least square problem.
Second, IRLS is not a greedy method.

Moreover, we conduct extensive simulations to demonstrate that our method can greatly improve OMP \cite{Tropp2007}, Subspace pursuit (SP) \cite{Dai2009}, and OMPR \cite{Jain2011} in terms of memory usage and computation cost for large-scale problems.


\subsection{Outline of This Paper}\label{Sec: Outline}
The rest of this paper is organized as follows.
In Sec. \ref{sec:Problem Formulation}, we describe the bottleneck of current greedy algorithms that is to solve the least square problem.
The proposed idea of fast and cost-effective least square solver for speeding greedy approaches along with theoretical analysis is discussed in Sec. \ref{sec:Proposed Method}.
In Sec. \ref{Sec: Experimental Results}, extensive simulations are conducted to show that our method indeed can be readily incorporated with state-of-the-art greedy algorithms, including OMP, SP, and OMPR, to improve their performance in terms of the memory and computation costs.
Finally, conclusions are drawn in Sec. \ref{sec:conl}.

\section{Problem Statement}\label{sec:Problem Formulation}

In this paper, without loss of generality, we focus on a signal $\in \mathbb{R}^{N_1\times N_2\times\cdots\times N_D}$, where $N=N_1$ and $N=N_1\times N_2$ are large enough with respect to $D=1$ and $D=2$.
When $D=2$, the signal usually is reshaped to $1$-D form in the context of compressive sensing.
Ideally, the memory cost of a compressive sensing algorithm should be $O(N)$, which is the minimum requirement for saving the original signal, $x$.
The computational cost, however, depends on an algorithm itself.
Since greedy algorithms share the same framework composed of support detection and solving least square problem, we shall focus on reducing the costs of these two procedures.

In this section, we discuss the core of proposed fast and cost-effective greedy approach and take OMP as an example for subsequent explanations.
We shall point out the dilemma in terms of memory cost and computation cost when handling large-scale signals. In fact, both costs suffer from solving the least square problem, which cannot be conducted by operator directly.

First, we follow the notations mentioned in the previous section and briefly introduce OMP \cite{Tropp2007} in a step-by-step manner as follows.
\begin{enumerate}
\item Initialize the residual measurement $r_{0}=y$ and initialize the set of selected supports $S_{0}=\{\}$.
Let the initial iteration counter be $i=1$.
Let $A_{S}$ be the sub-matrix of $A$, where $A_{S}$ consists of the column of $A$ with indices belonging to the support set $S$. $A^{T}$ is the transpose of $A$.
\item Detect supports (or positions of significant components) by seeking maximum correlation from
\begin{equation}
t=\arg\!\max_{t} | \left(A^{T}r_{i-1}\right)_{t} |,
\label{eq:OMP maximize correlation}
\end{equation}
and update the support set $S_{i}=S_{i-1}\cup \{t\}$.
\item Solve a least square problem $s_{i}=(A_{S_{i}}^{T} A_{S_{i}})^{-1} A_{S_{i}}^{T} y$, and
update residual measurement  $r_{i}=y-A_{S_{i}}s_{i}$.
\item If $i=K$, stop; otherwise, $i=i+1$ and return to Step 2.
\end{enumerate}

Tropp and Gilbert \cite{Tropp2007} derive that the computational complexity of OMP is bounded by Step 2 (support detection) with $O(MN)$ and Step 3 (solving least square problem) with $O(MK)$, and the memory cost is $O(MN)$ when $A$ is executed in a matrix form.
In this paper, we call it matrix-based OMP (M-OMP).
As mentioned in Sec. \ref{Sec: RW}, $A$ can be designed to be an SRM conducted by operator.
Nevertheless, operator is only helpful for certain operations such as $A$ and $A^{T}$.
For example, if $A$ is a partial random Fourier matrix, $Ax = \mathcal{D}\left( FFT(x)\right)$ and $A^{T}y = IFFT(\hat{y})$ can be quickly calculated, where $\hat{y}=[ y^{T},\underbrace{0,...,0}_{N-M} ]^{T}$ and $IFFT(\cdot)$ denotes inverse FFT function.
In this paper, we call it operator-based OMP (O-OMP).

Unfortunately, the key is that $(A_{S_{i}}^{T}A_{S_{i}})^{-1}$ still cannot be quickly computed in terms of operator.
Hence, it requires $O(KM)$ to store $A_{S_{i}}$.
Instead of calculating $(A_{S_{i}}^{T}A_{S_{i}})^{-1}$ directly, by preserving the Cholesky factorization of $(A_{S_{i-1}}^{T}A_{S_{i-1}})^{-1}$ at the $(i-1)^{th}$ iteration for subsequent use, Step 3 is accelerated and the memory cost is reduced to $O(K^2)$.
Moreover, it is worth mentioning that the sparsity $K$ of natural signals is often linear to signal length $N$.
For example, the number of significant DCT coefficients for an image usually ranges from $0.01N$ to $0.1N$.
Also CS has shown that $M$ must be linear to $K\log N$ for successful recovery with high probability.
Under the circumstance, $O(N)$ is equivalent to $O(K)$ and $O(M)$ in the sense of big-O notation.
We can see that when $N$ is large, $O(K^2)$ dominates the memory cost because $K^2 \gg N$.
Thus, Step 3 makes OMP infeasible for recovering large-scale signals.

In fact, greedy algorithms share the same operations, {\em i.e.}, $A^{T}r_{i-1}$ in Step 2 and $(A_{S_{i}}^{T} A_{S_{i}})^{-1} A_{S_{i}}^{T}y$ in Step 3, where the main difference is that the support set $S_{i}$ is found by different ways, and face the same dilemma.
A simple experiment is conducted and results are shown in Fig. \ref{fig:MemoryCost} to illustrate the comparison of memory usage among M-OMP, O-OMP, and ideal cost (which is defined as $N \times$ 8 bytes required for Double data type in Matlab).
The OMP code running in Matlab was downloaded from SparseLab (http://sparselab.stanford.edu).
Obviously, though the memory cost of O-OMP is reduced without storing $A$, it still far higher than that of ideal cost.
It is also observed that both M-OMP and O-OMP exhibit the same slope. Specifically, M-OMP and O-OMP cost $O(MN)$ and $O(K^{2})$, respectively.
As mentioned before, since $M,K$ are linear to $N$, it means $O(K^{2})=O(N^{2})$ and $O(MN)=O(N^{2})$ such that both orders of memory cost of  M-OMP and O-OMP  are the same and are larger than that of ideal case.
Consequently, solving the least square problem becomes a bottleneck in greedy approaches.
This challenging issue will be solved in this paper.

\begin{figure}[h]
  \centering
  \centering{\epsfig{figure=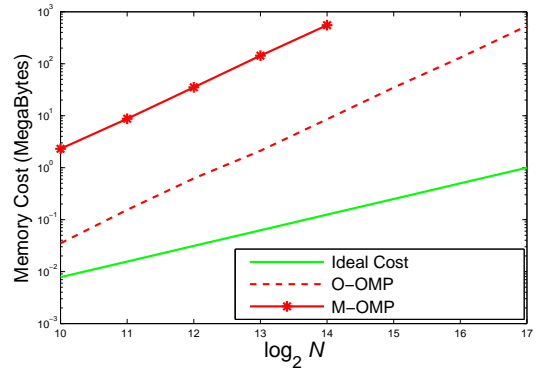,width=8cm}}
\hfill
\caption{Memory cost comparison among the matrix-based OMP, operator-based OMP, and ideal cost under $M=\frac{N}{4}$ and $K=\frac{M}{4}$. }
\label{fig:MemoryCost}
\end{figure}

\section{Proposed Method for Speeding Recovery of Greedy Algorithms}\label{sec:Proposed Method}
In this section, we first introduce how to determine a sensing matrix, which can be easily implemented by operator.
Then, we reformulate the least square problem as a weighted one, which is solved by conjugate gradient (CG) method to avoid involving the sub-matrices, $A_{S_{i}}$ or $A_{S_{i}}^{\dag}$.
We first prove that the solutions to the least square problem and its weighted counterpart are the same, and then prove that the solutions to the weighted least square problem and its CG-based counterpart are the same.

\subsection{Sensing Matrix}\label{sec: Sensing Matrix Design}
A random Gaussian matrix is commonly used as the sensing matrix as it and any orthonomal basis can pair together to satisfy RIP and MIP in the context of compressive sensing.
The use of random Gaussian matrix as the sensing matrix, however, leads to the overhead of storage and computation costs.
Although storage consumption can be overcome by using a seed to generate a random Gaussian matrix, it still encounters high computational cost.

In \cite{Do2012}, Do {\em et al.}  propose a framework, called Structurally Random Matrix (SRM), defined as:
\begin{equation}
	\Phi = DFR,
\label{DFR}
\end{equation}
where $D \in\mathbb{R}^{M \times N}$ is a sampling matrix, $F \in\mathbb{R}^{N \times N}$ is an orthonormal matrix, and $R \in\mathbb{R}^{N \times N}$ is a uniform random permutation matrix (randomizer).
Since the distributions between a random Gaussian matrix and SRM's $\Phi$ are verified to be similar, we choose Eq. (\ref{DFR}) as the sensing matrix for our use.
It should be noted that $D$ and $R$ can inherently be replaced by operators but $F$ depends on what kind of orthonormal basis is used.
It is obvious that any fast transform can be adopted as $F$.
In our paper, we set $F$ to the Discrete Cosine Transform (DCT) due to its fast computation and cost-effectiveness.
There are literatures discussing the design of sensing matrix but it is not the focus of our study here.

\subsection{Reformulating Least Square Problem: Weighted Least Square Problem} \label{sec: Weighted least square}
In Sec. \ref{sec:Problem Formulation}, we describe that the bottleneck of greedy algorithms is the least square problem.
To solve the problem, it is reformulated as a weighted least square problem in our method.
To do that, we first introduce a weighted matrix $W\in\mathbb{R}^{N \times N}$ defined as:
\begin{equation}
	W_{S_{i}}[j,j] =\left \{ \begin{aligned} 1, j \in S_{i}  \\ 0, j \not\in S_{i} \end{aligned} \right.,
\label{eq: W}
\end{equation}
where, without loss of generality, $S_{i}=\{1,2,...,i\}$ denotes a support set at the $i$-th iteration and $W_{S_{i}}[j,j]$ is the $(j,j)^{th}$ entry of $W_{S_{i}}$.
As can be seen in Eq. (\ref{eq: W}), the weighte matrix $W$ is designed to make the supports unchanged.

Then, we prove that both solutions of the least square problem and its weighted counterpart are the same.
\begin{theorem}
\label{theorem:equivalent form}
Suppose the sub-matrix $ A_{S_{i}} \in \mathbb{R}^{M \times K}$ of $A$ has full column rank with support set $S_i$.
Let $s_{i} \in \mathbb{R}^{K}$ be the solution to the least square problem:
\begin{equation}
    	s_{i} = \arg\!\min_{\hat{s}} \| y - A_{S_{i}}\hat{s} \|_{2},
    \label{eq: minimize least square equation}
\end{equation}
and let $\theta_{i} \in \mathbb{R}^{N}$ be the solution to the weighted least square problem:
  \begin{equation}
    	\theta_{i} = \arg\!\min_{\hat{\theta}} \| y - AW_{S_{i}}\hat{\theta} \|_{2}.
  \label{eq: wighted minimize least square equation}
\end{equation}
We have
$$ s_{i}[j] = \theta_{i}[j], \quad \text{for} \quad 1\leq j\leq i.$$
\end{theorem}
\begin{proof}
Let $\theta_{*} = [s_{i}^{T}\ 0]^{T} =   \begin{bmatrix}
       \left(A_{S_{i}}^{T}A_{S_{i}}\right)^{-1} A_{S_{i}}^{T}    \\
    0
\end{bmatrix} y$ be a solution minimizing $\| y - AW_{S_{i}}\theta_{*} \|_{2}$. Then $\{ \theta_{*} + v| v \in Null(AW_{S_{i}}) \}$ is the solution set of Eq. (\ref{eq: wighted minimize least square equation}).
Since $A_{S_{i}}$ has full column rank with $rank(A_{S_{i}})=K$ and $Null(AW_{S_{i}})=span\left(e_{i+1},e_{i+2},...,e_{N}\right)$, where $e_{i}$ is a standard basis.
Thus, no matter what $v$ is,  the first $i$ entries of $\theta_{*} + v$ are invariant. We complete the proof.
\end{proof}

\subsection{Reformulating Weighted Least Square as CG-based Weighted Least Square Problem} \label{sec: CG-Weighted least square}
We can see from Eq. (\ref{eq: wighted minimize least square equation}) that the introduce of weighted matrix involves $AW_{S_{i}}$ instead of submatrix $A_{S_{i}}$ so that $A$ can be calculated by fast operator.
Nonetheless, the closed-form solution of Eq. (\ref{eq: wighted minimize least square equation}) is
$  \begin{bmatrix}
       \left(A_{S_{i}}^{T}A_{S_{i}}\right)^{-1} A_{S_{i}}^{T}    \\
    0
\end{bmatrix} y$ and
still faces the difficulty in that $ \left(A_{S_{i}}^{T}A_{S_{i}}\right)^{-1} A_{S_{i}}^{T}  $ cannot be easily implemented by operator.
Instead of seeking closed-form solutions, we aim to explore the first-order methods ({\em e.g.}, gradient descent), which have the following advantages:
(1) the operations only involve $A$ and $A^{T}$ instead of the pseudo-inverse of $A$,
(2) the convergence rate is nearly dimension-independent \cite{Cevher2014}, and
(3) if $A$ is a sparse matrix, the computation cost can be further reduced.
Conjugate gradient (CG) \cite{Sluis1986}\cite{Kaasschieter1988} is a well-known first order method to numerically approximate the solution of symmetric, positive-definite or positive-semidefinite system of linear equations.
Thus, CG benefits from the advantages of the first-order method.
However, the matrix $AW_{S_{i}}$ in Eq. (\ref{eq: wighted minimize least square equation}) is not symmetric.
Thus, we reformulate Eq. (\ref{eq: wighted minimize least square equation}) in terms of CG as follows.
We will prove that both the solutions to the weighted least square problem and CG-based weighted least square problem are the same.

\begin{theorem}
\label{theorem:2nd equivalent form }
Suppose the sub-matrix $ A_{S_{i}} \in \mathbb{R}^{M \times K}$ of $A$ has full column rank with support set $S_{i}$.
Let $\theta_{i} \in \mathbb{R}^{N}$ be the solution to weighted least square problem defined in Eq. (\ref{eq: wighted minimize least square equation}).
Let $\widetilde{\theta} \in \mathbb{R}^{N}$ be the solution to the CG-based weighted least square problem reformulated from Eq. (\ref{eq: wighted minimize least square equation}) as:
\begin{equation}
    	\widetilde{\theta}_{i} = \arg\!\min_{\hat{\theta}} \| W_{S_{i}}^{T}A^{T}(y - AW_{S_{i}}\hat{\theta}) \|_{2}.
    \label{eq:wighted minimize least square equation 2nd}
\end{equation} 
Then
$$ \theta_{i}[j] =\widetilde{\theta}_{i}[j], \quad \text{for} \quad 1\leq j\leq i.$$
\end{theorem}
\begin{proof}
We first note that now $W_{S_{i}}^{T}A^{T}AW_{S_{i}}$ is symmetirc to meet the requirement of CG.
Since $  \begin{bmatrix}
         \left(A_{S_{i}}^{T}A_{S_{i}}\right)^{-1} A_{S_{i}}^{T}      \\
    0
\end{bmatrix} y$ is an optimal solution for both Eq. (\ref{eq: wighted minimize least square equation}) and Eq. (\ref{eq:wighted minimize least square equation 2nd}), the solution set of  Eq. (\ref{eq:wighted minimize least square equation 2nd}) can be expressed as:
$$\{   \begin{bmatrix}
         \left(A_{S_{i}}^{T}A_{S_{i}}\right)^{-1} A_{S_{i}}^{T}      \\
    0
\end{bmatrix} y + v | AW_{S_{i}}v \in Null( W_{S_{i}}^{T}A^{T})  \}.$$
It should be noted that $Null( W_{S_{i}}^{T}A^{T}) = Null( A_{S_{i}}^{T}) $.
In addition,  $AW_{S_{i}}v \in \mathcal{C}(AW_{S_{i}})$, where $\mathcal{C}(AW_{S_{i}})$ denotes the column space of $AW_{S_{i}}$ and $ \mathcal{C}(AW_{S_{i}}) =  \mathcal{C}(A_{S_{i}})$.
Since $\mathcal{C}(A_{S_{i}}) \bigcap Null( A_{S_{i}}^{T}) = \{0\}$, it implies $AW_{S_{i}}v = 0$. In other words, $v \in Null(AW_{S_{i}})$. Due to $Null(AW_{S_{i}})=span\left(e_{i+1},e_{i+2},...,e_{N}\right)$, the first $i$ entries of $\begin{bmatrix}
       \left(A_{S_{i}}^{T}A_{S_{i}}\right)^{-1} A_{S_{i}}^{T}    \\
    0
\end{bmatrix}y + v$ are invariant.
Similarly, the solution set of  Eq. (\ref{eq: wighted minimize least square equation}) is  $$\{   \begin{bmatrix}
       \left(A_{S_{i}}^{T}A_{S_{i}}\right)^{-1} A_{S_{i}}^{T}    \\
    0
\end{bmatrix} y + v | v \in Null( AW_{S_{i}})  \}.$$ The first $i$ entries of $\begin{bmatrix}
         \left(A_{S_{i}}^{T}A_{S_{i}}\right)^{-1} A_{S_{i}}^{T}      \\
    0
\end{bmatrix} y  + v$ are also invariant. We complete the proof.
\end{proof}

So far, we prove that, with correct support detection, the optimal solution to the CG-based weighted least square problem in Eq. (\ref{eq:wighted minimize least square equation 2nd}) is equivalent to that to the original least square problem in Eq. (\ref{eq: minimize least square equation}).
In addition, the matrix $W_{S_{i}}^{T}A^{T}AW_{S_{i}}$  in Eq. (\ref{eq:wighted minimize least square equation 2nd}) is symmetric and can quickly be solved by CG method.
 Nevertheless, \cite{Kaasschieter1988} points out that the system of linear equations with a positive-semidefinite matrix diverges unless some conditions are satisfied.
Thus, Theorem \ref{theorem:convergence of CG} further shows the condition of convergence.

\begin{theorem} \cite{Kaasschieter1988}
If $W_{S_{i}}^{T}A^{T}y  \in \mathcal{C}(W_{S_{i}}^{T}A^{T}AW_{S_{i}})$, CG method converges but the solution is not unique.
\label{theorem:convergence of CG}
\end{theorem}

Now, we check whether the CG-based weighted least square problem converges.
Again, let $S_{i}=\{1,2,...,i\}$ denote a support set.
Then, we have $W_{S_{i}}^{T}A^{T}y = [ y^{T}A_{S_{i}} \ 0]^{T}$ and $$\mathcal{C}(W_{S_{i}}^{T}A^{T}AW_{S_{i}}) = \mathcal{C}(  \begin{bmatrix}
       A_{S_{i}}^{T} A_{S_{i}}    \\
    0
\end{bmatrix}) .$$
Because $ A_{S_{i}}^{T} A_{S_{i}} \in \mathbb{R}^{i \times i}$ is full rank, the first $i$ entries of $ W_{S_{i}}^{T}A^{T}y$ must be spanned by the basis of $ A_{S_{i}}^{T} A_{S_{i}}$.
The remaining $N-i$ entries are $0$ and trivial.
Thus, it implies the CG part of our CG-based weighted least square solver satisfies Theorem \ref{theorem:convergence of CG} and converges.

\subsection{Speeding Orthogonal Matching Pursuit and Complexity Analysis} \label{sec: OP CG}
In the previous section, we descirbe the proposed CG-based weighted least square solver.
In this section, we first show that the CG-based solver can be implemented easily by operator.
Then, we combine it with OMP as a new paradigm to achieve fast OMP.
Moreover, we discuss the convergence rate of CG and derive the computation complexity of proposed operator-based OMP via CG (dubbed as CG-OMP).
Finally, we conclude that the memory cost of CG-OMP achieves ideal $O(N)$ if $\Psi$ is also conducted by operator.

\begin{algorithm}[!t]
\small
\centering
\setlength{\abovecaptionskip}{0pt}
\setlength{\belowcaptionskip}{0pt}
\caption{Proposed Orthogonal Matching Pursuit}
\label{alg:Operator-based OMP}
\begin{tabular}[t]{p{8.4cm}l}
\hline\\
\textbf{Input:} $\bm{y}$, $A$, $K$;\quad \textbf{Output:} $s_{K}$; \\
\textbf{Initialization:} $i=1$, $r_{0}=y$,$s_{0}=A^{T}y$, $S_{0}=\{ \}$; \\
\hline\hline\\
01. \textbf{function} \textbf{Proposed Operator-based OMP}()\\
02. \quad\textbf{for} $i = 1$ to $K$ \\
03. \quad\quad $t = \arg\!\max_{\hat{t}} \left| (A^{T}r_{i-1})_{\hat{t}}  \right|$;\\
04. \quad\quad $S_{i} = S_{i-1} + t  $;\\
05. \quad\quad Assign $W_{S_{i}}$ according to Eq. (\ref{eq: W});\\
06. \quad\quad $s_{i} =$\textbf{CG}$\left( A,W_{S_{i}},y,i \right)$;     \\
07. \quad\quad $r_{i} = y - As_{i}$;     \\
07. \quad\quad $i=i+1$;     \\
08. \quad\textbf{end for}\\
09. \quad\textbf{Return}:$s_{K}$;\\
10. \textbf{end} \textbf{function}\\
11. \textbf{function} $\left[ \ \hat{s} \ \right]$=\textbf{CG}$\left( A,W_{S_{i}},y,i \right)$\\
12. \quad $b = W_{S_{i}}^{T}A^{T}y$, $H = W_{S_{i}}^{T}A^{T}AW_{S_{i}}$;\\
13. \quad $d_{0}=r_{0}=b$, $\hat{s}=0$;\\
14. \quad $j=0$;\\
15. \quad\textbf{while}($\| r_{j} \|_{2} \leq \xi $) \\
16. \quad\quad $\alpha_{j}= \frac{r_{j}^{T}r_{j}}{d^{T}_{j}Hd_{j}}$;\\
17. \quad\quad $\hat{s} = \hat{s} + \alpha_{j}d_{j}  $;\\
18. \quad\quad $r_{j+1} = r_{j} - \alpha_{j}Hd_{j}  $;\\
19. \quad\quad $\beta_{j+1} = \frac{r_{j+1}^{T}r_{j+1}}{r_{j}^{T}r_{j}}  $;\\
20. \quad\quad $d_{j+1} = r_{j+1}+\beta_{j+1}d_{j}  $;\\
21. \quad\quad $j = j+1$;\\
22. \quad\textbf{end while}\\
23. \quad\textbf{Return}:$\hat{s}$;\\
24. \textbf{end} \textbf{function}\\
\hline
\end{tabular}
\end{algorithm}
\renewcommand\arraystretch{1}

Algorithm \ref{alg:Operator-based OMP} describes the proposed CG-OMP method (Lines 01 - 10), which employs a CG technique (Lines 11 - 23).
It is worth mentioning that $\xi$ in Line 15 controls the precision of CG method. If $\xi = 0$, it means exact precision such that the output of CG method is equal to least square solution. If $\xi >0$, CG method attains finite precision. However, the result with finite precision is not certainly worse than that with exact precision especially under noisy interference, as later discussed in the 4-th paragraph of Sec. \ref{Sec: CCC}.

Now we check whether all matrix operations can be implemented by operators in Algorithm \ref{alg:Operator-based OMP} in the following.
\begin{itemize}
\item $W_{S_{i}} \in \mathbb{R}^{N \times N}$:\\
$W_{S_{i}}$ is a diagonal matrix. Thus, $W_{S_{i}}x$ is equal to assign $x[j]=0$ for $j \not\in S_{i}$. The memory cost is $O(K)$ to store the indices of support set and the computation cost is $O(N)$.
\item $\Phi = DFR \in \mathbb{R}^{M \times N} $:\\
$Dx$ is equal to randomly choose $M$ indices from $N$ entries in $x$.
$Fx=DCT(x)$, where DCT can be speeded by FFT.
$Rx$ is equal to randomly permute the indices of vector $x$.
The memory cost is bounded by $O(N)$ in order to store the sequence of random permutation and the computation cost is $O(N \log N)$.
\item $\Psi \in \mathbb{R}^{N \times N} $:\\
If $\Psi$ is a deterministic matrix, it is not necessarily to be stored.
Thus, the memory cost of $\Psi x$ is $O(N)$.
The computational cost depends on whether $\Psi x$ can be speeded up.
In the worst case, it costs $O( N^2)$.
\end{itemize}
Other operations such as $r_{j}^{T}r_{j}$ (Line 16) only involve multiplications between vectors.
Thus, both memory cost and computation cost are bounded by $O(N)$.

From the above analysis, one can see that the computation cost of Algorithm \ref{alg:Operator-based OMP} is mainly bounded by $\Psi$.
We discuss some applications below, where the computation cost involving $\Psi$ is low.
For compressive sensing of images, wavelet transform is often chosen as $\Psi$ such that $\Psi x$ costs $O(N)$.
For magnetic resonance imaging (MRI), partial Fourier transform is selected as the component $F$ of the sensing matrix expressed as $DFR$.
In this case, $\Psi$ is $I$ in order to satisfy MIP or RIP, and costs nothing.
Spectrum sensing is another application, where $\Psi$ is a discrete Fourier transform matrix done with $O(N \log N)$.

Now, the total cost of Algorithm \ref{alg:Operator-based OMP} is discussed.
Both the computation and memory costs of Lines 1-10 except Line 6 (CG method) will be $O(N \log N)$ and $O(N)$, respectively.
As for the memory cost of CG method, it needs $O(N)$.
Therefore, the total memory cost of Algorithm \ref{alg:Operator-based OMP} is bounded by $O(N)$.

In addition, the computation cost is related to two factors, {\em i.e.}, the number of iterations to converge in CG and $\Psi$.
They are further discussed as follows.

\begin{theorem}
\label{theorem: num of iteration to converge} (Theorem 2.2.3 in \cite{Kelley1995})
Let $H$ be symmetric and positive-definite.
Assume that there are exactly $k < N$ distinct eigenvalues in $H$.
Then, CG terminates in at most $k$ iterations.
\end{theorem}
In our case, $H = W_{S_{i}}^{T}A^{T}AW_{S_{i}}$  is positive-semidefinite instead of positive-definite. Thus, we derive the following theorem.
\begin{theorem}
\label{theorem: num of iteration to converge semi ver}
Given $b = W_{S_{i}}^{T}A^{T}y$ and $H = W_{S_{i}}^{T}A^{T}AW_{S_{i}}$, solving Eq. (\ref{eq:wighted minimize least square equation 2nd}) requires the number of iterations at most $K$ in CG, where $K$ is the sparsity of an original signal.
\end{theorem}
\begin{proof}
Without loss of generality, let support set $S_{i}=\{1,2,...,i\}$. We start from another optimization problem:
 \begin{equation}
 \label{eq: CG for spd}
\bar{s}_{i} = \arg\!\min_{\hat{s}} \| A_{S_{i}}^{T}(y -  A_{S_{i}}\hat{s}) \|_{2}.
 \end{equation}
Following the same skill in Theorem \ref{theorem:2nd equivalent form }, $ \bar{s}_{i}$ is a unique and optimal solution to both Eq. (\ref{eq: minimize least square equation}) and Eq. (\ref{eq: CG for spd}). Thus, $ \bar{s}_{i}$ is also the solution of  Eq. (\ref{eq:wighted minimize least square equation 2nd}) for the first $i$ entries.
Furthermore, $ A_{S_{i}}^{T}A_{S_{i}}$ is non-singular such that $ A_{S_{i}}^{T}A_{S_{i}}$ is a positive-definite matrix and has at most $i$ distinct eigenvalues.
From Theorem \ref{theorem: num of iteration to converge}, solving Eq. (\ref{eq: CG for spd}) requires at most $i$ iterations.
Then, we have $b = W_{S_{i}}^{T}A^{T}y = [ y^{T}A_{S_{i}}\ 0]^{T}$ and  $H =  W_{S_{i}}^{T}A^{T}AW_{S_{i}} = \left[
\begin{aligned}
       A_{S_{i}}^{T} A_{S_{i}} \quad  0  \\
    0           \quad     \quad\quad        0
\end{aligned}
\right] $
in Eq.  (\ref{eq:wighted minimize least square equation 2nd}).
When we only take the first $i$ entries of $b$ and the left-top $i \times i $ submatrix of $H$ into consideration, it is equivalent to solving Eq.  (\ref{eq: CG for spd}).
This fact can be checked trivially by comparing each step of CG  for both optimization problems in Eq.  (\ref{eq:wighted minimize least square equation 2nd}) and Eq. (\ref{eq: CG for spd}).
Thus, the first $i$ entries of $\hat{s}$ in Eq. (\ref{eq:wighted minimize least square equation 2nd}) is updated in the same manner with that of $\bar{s}$ in Eq.  (\ref{eq: CG for spd}).
The remaining $N-i$ entries of $\hat{s}$ are unrelated to convergence because the $N-i$ entries of $H\hat{s}$ are zero. In sum, the required number of iterations to converge in Eq.  (\ref{eq:wighted minimize least square equation 2nd}) is identical to that in Eq. (\ref{eq: CG for spd}).
Hence, solving Eq. (\ref{eq:wighted minimize least square equation 2nd}) also requires at most $i$ iterations.
Since $i \leq K$, the number of iterations is at most $K$.
We complete the proof.
\end{proof}

Moreover, for each iteration in CG, the operation, $Hd_{j}$, dominates the whole computation cost.
It is obvious that if $\Psi$ can be executed with $O(N \log N)$ or even lower computation complexity, $Hd_{j}$ costs $O(N \log N)$ since $H$ involves $\Psi$, which spends $O(N \log N)$.
Note that the total computation complexity of CG-OMP will be $O(NK^{2}\log N)$, where $K^2$ comes from the outer loop in OMP, which needs $O(K)$, and solving Eq. (\ref{eq:wighted minimize least square equation 2nd}) that requires the number of iterations at most $K$ in CG, as proved in Theorem \ref{theorem: num of iteration to converge semi ver}.
On the other hand, if, under the worst case, $\Psi$ costs $O(N^{2})$ operations, the computation complexity of CG-OMP will be $O(N^{2}K^{2})$.
For applications that accept finite-precision accuracy instead of exact precision, CG \cite{Sluis1986} requires fewer steps ($\leq K$) to achieve approximation.
Under the circumstance, the complexity of CG-OMP nearly approximates $O(NK \log N)$ and $O(N^{2}K)$ operations for $\Psi$ with complexity $O(N \log N)$ and $O(N^{2})$, respectively.

Consequently, a reformulation for solving a least square problem is proposed based on CG such that the new matching pursuit methodology can deal wtih large-scale signals quickly.
It should be noted that, in the future, CG may be substituted with other first order methods that outperform CG.
The proposed idea can also be readily applied to other greedy algorithms to enhance their performance.

\subsection{Strategies for Reducing the Cost of $\Psi$} \label{secc:separable psi}
We further consider that if $\Psi$ is a learned dictionary, it will become a bottleneck for operator-based algorithms since it requires $O(N^{2})$ for storage.
To overcome this difficulty, $\Psi$ should be learned in a tensor structure.
Let $x = vec(X)$, where $X \in \mathbb{R}^{\sqrt{N} \times \sqrt{N}}$  and $vec(\cdot)$ is a vectorization operator.
That is, a two-dimensional vector is reshaped to a one-dimensional vector.
Then, we can learn a 2D dictionary such that $x= \Psi s = vec( \Psi_{1} S\Psi_{2}^{T}) $ with $s = vec(S)$ and $\Psi =  \Psi_{2} \otimes \Psi_{1}$ ($\otimes$ is a Kronecker product).
Under the circumstance, all operations in CG involving $\Phi\Psi s$ can be replaced by $\Phi vec(\Psi_{1}S\Psi_{2}^{T} )$.
Moreover, both $\Psi_{1}$ and $\Psi_{2}$ only require $O(N)$ in terms of memory cost.
In the literature, the existing algorithms for 2D separable dictionary learning include \cite{shi2013}\cite{Hawe2013}\cite{Hsieh2014}.

\section{Experimental Results}\label{Sec: Experimental Results}
In this section, we conduct comparisons among O-OMP, M-OMP, and CG-OMP in terms of the memory cost and computation cost.
The code of OMP was downloaded from SparseLab (http://sparselab.stanford.edu).
We have also applied the proposed fast and cost-effective least square solver to SP and OMPR in order to verify if our idea can speed the family of matching pursuit algorithms.
For SP and OMPR, we implemented the corresponding original matrix-based versions  (M-OMPR and M-SP), original operator-based versions  (O-OMPR and O-SP), and proposed CG-based (CG-OMPR and CG-SP).
It should be noted that although both SP and OMPR work well for increasingly adding a index to the support set like OMP, they are not accelerated by Cholesky factorization.

\subsection{Simulation Setting}
The simulations were conducted in an Matlab R2012b environment with an Intel CPU Q6600 and $4$ GB RAM under Microsoft Win7 ($64$ bits).

The model for the measurement vector in CS is $y=As+\eta$, where $\eta$ is an addictive Gaussian noise with standard deviation $\sigma_{\eta}$.
The input signal $s$ was produced via a Gaussian models as:
\begin{equation}
  s \sim pN\left( 0,\sigma_{on}^{2} \right),
  \label{eq:signalproduce}
\end{equation}
which was also adopted in \cite{Mohimani2009}.
In Eq. (\ref{eq:signalproduce}), $p$ is the probability of the activity of a signal and controls the number of non-zero entries of $x$.
Sparsity $K$ is defined to be $K=pN$.
$\sigma_{on}$ is standard deviation for input signal.
$\Phi$ is designed from SRM and $\Psi$ is chosen to be a discrete cosine transform.
In the following experiments, $M=\frac{N}{4}$, $K=\frac{M}{4}$, $p=0.0625$, $\sigma_{on}=1$, and $\sigma_{\eta}=0.01$.

\subsection{Memory Cost Comparison}
Fig. \ref{fig:MemoryCost_2} shows the comparison in terms of memory cost vs. signal length $N$.
Since the matrix-based algorithms (M-OMP, M-OMPR, and M-SP) run out of memory, their results are not shown in Fig. \ref{fig:MemoryCost_2} (note that the result regarding matrix-based OMP can be found in Fig. \ref{fig:MemoryCost}).
First, we can observe from Fig. \ref{fig:MemoryCost_2} that O-OMP, O-OMPR, and O-SP still require about $O(N^2)$ and fail to work when $N > 2^{18}$.
Second, in contrast with O-OMP, O-OMPR, and O-SP although CG-OMP, CG-OMPR, and CG-SP need more memory costs than the ideal cost, which is $O(N)$, their slopes are nearly identical, which seems to imply that the CG-based versions incur larger Big-O constants.
Thus, the proposed idea of fast and cost-effective least square solver is readily incorporated with the existing greedy algorithms to improve their capability of handling large-scale signals.

\begin{figure}[h]
  \centering
  \centering{\epsfig{figure=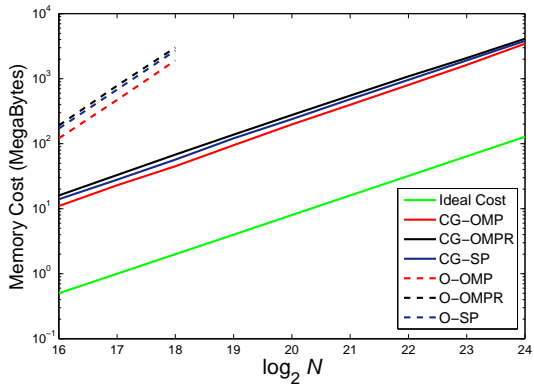,width=8cm}}
\hfill
\caption{Memory cost vs. Signal length, where  $M=\frac{N}{4}$ and $P=\frac{M}{4}$ ($p=0.0625$). }
\label{fig:MemoryCost_2}
\end{figure}

\subsection{Computation Cost Comparison}\label{Sec: CCC}
Before illustrating the computation cost comparison, we first discuss the convergence condition of CG in Algorithm \ref{alg:Operator-based OMP} as follows:
1) For exact precision, $\| r_{j} \|_{2} =0$ is set as the stopping criterion.
As described in Theorem \ref{theorem: num of iteration to converge semi ver}, it costs at most $K$ iterations to converge.
2) With inexact or finite precision, we set $\| r_{j} \|_{2} \leq \xi$ with $\xi > 0$.
Though the precision is finite, the signals, in practice, are interfered with noises and solutions with finite precision are adequate.
Under the circumstance, the number of iterations required to converge will be significantly decreased.

Figs. \ref{fig:ComputationCost_1}(a), (b), and (c) show the computation cost vs. signal length for OMP, OMPR, and SP, respectively, under the condition that the precision of CG was set to be exact.
In other words, we fix the same reconstruction quality for all comparisons in Fig. \ref{fig:ComputationCost_1} and discuss the computation costs for these different versions of algorithms.
It should be noted that some curves are cut because of running out of memory.

From Fig. \ref{fig:ComputationCost_1}, it is observed that the operator-based strategy (denoted with dash curves) or our proposed CG-based method (denoted with solid curves) can effectively reduce the order of computation cost in comparison with the matrix-based strategy (denoted with solid-star curves).
More specifically, in  Fig. \ref{fig:ComputationCost_1}(a), it is noted that M-OMP is only fast than CG-OMP with $N \leq 2^{14}$ due to smaller Big-O constant.
However, CG-OMP outperforms M-OMP in the end since the order of computation complexity of CG-OMP is lower than that of M-OMP.
In particular,  such improvements are significant for large-scale signals (with large $N$).
Moreover, in Figs. \ref{fig:ComputationCost_1}(b) and (c), O-OMPR and O-SP have the same orders with M-OMPR and M-SP because no Cholesky factorization is used.

On the other hand, when finite precision is considered, we consider two cases of setting $\xi = 10^{-5}$ and $\xi = 10^{-10}$ to verify that finite precision is adequate under the condition of noisy interferences.
Taking exact precision as the baseline, the difference of SNR values for reconstruction between settings for exact precision and $\xi = 10^{-5}$ is about $\pm 0.1$ dB.
Similarly, the difference between exact precision and $\xi = 10^{-10}$ is about $\pm 0.01$ dB.
Tables \ref{Table: OMP precision}, \ref{Table: OMPR precision}, and \ref{Table: SP precision} further illustrate the comparisons of computation costs under different precisions.
The precision setting $\xi = 10^{-10}$ results in about four times faster than the exact precision but only sacrifices $\pm 0.1$ dB for the reconstruct quality, which is acceptable for many applications.
In fact, the performance occasionally is better because exact precision may lead to over-fitting.

\begin{figure}[h]
\begin{minipage}[b]{1\linewidth}
  \centering{\epsfig{figure=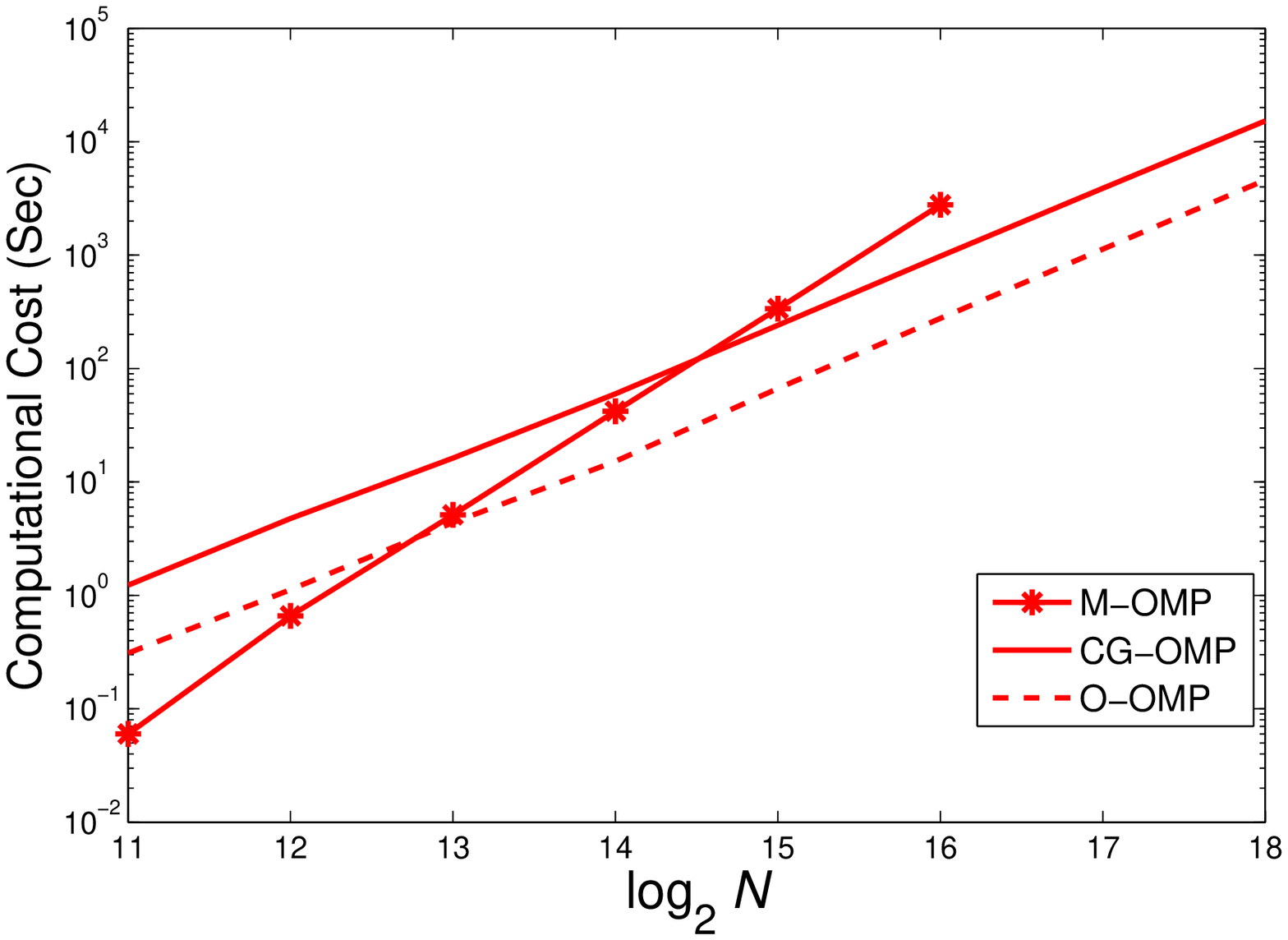,width=7.5 cm}}
  \centerline{(a) OMP}
\end{minipage}
\begin{minipage}[b]{1\linewidth}
  \centering{\epsfig{figure=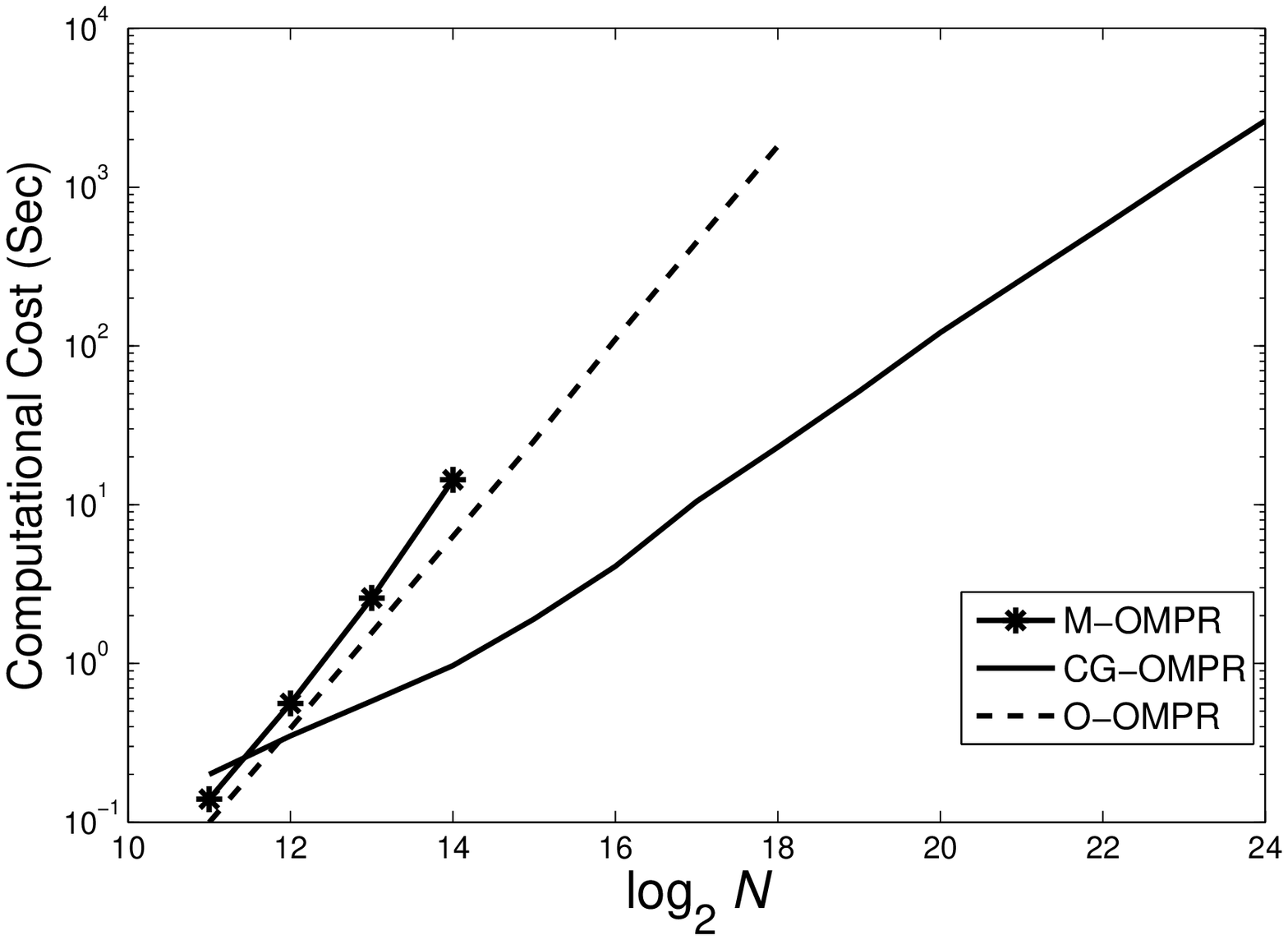,width=7.5 cm}}
  \centerline{(b) OMPR}
\end{minipage}
  \centerline{
\begin{minipage}[b]{1\linewidth}
  \centering{\epsfig{figure=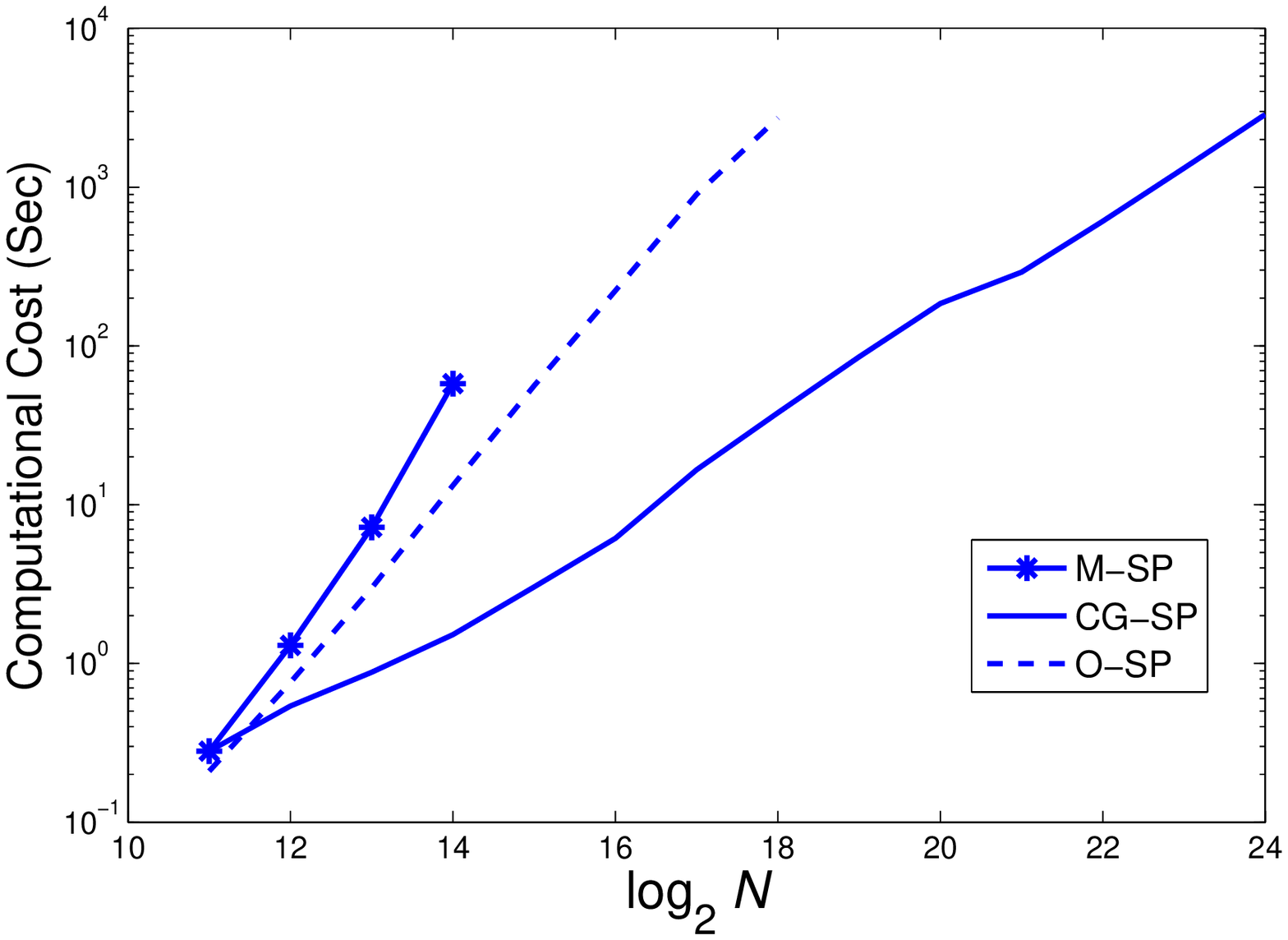,width=7.5 cm}}
  \centerline{(c) SP}
\end{minipage}
}
\hfill
\caption{Computation cost vs. signal length for OMP (a), OMPR (b), and SP (c), respectively, under $M=\frac{N}{4}$, and $K=\frac{M}{4} (p=0.0625)$.
The precision of CG was set to be exact.}
\label{fig:ComputationCost_1}
\end{figure}

\begin{table*}[t]
\scriptsize
\centering
\setlength{\abovecaptionskip}{0pt}
\setlength{\belowcaptionskip}{4pt}
\caption{The computational cost (in seconds) of CG-OMP under different precisions.}
\label{Table: OMP precision}
\doublerulesep=2pt
\begin{tabular}[tc]{|c||c|c|c|c|c|c|c|c|}
\hline
{\LARGE \textcolor{white}{o}}N& $2^{11}$& $2^{12}$& $2^{13}$ & $2^{14}$& $2^{15}$ & $2^{16}$ & $2^{17}$ & $2^{18}$  \\ \hline\hline
Exact precision & 1.2&  4.7 & 16.2 & 59.8& 241.0 & 980.3 & 3901.7 & 15301.4 \\ \hline
$\| r_{j} \|_{2} \leq 10^{-10}$ & 0.6& 2.3&  8.7& 35.5 & 131.3 & 564.7 & 2143.9 & 8341.3 \\ \hline
$\| r_{j} \|_{2} \leq 10^{-5}$ & 0.4& 1.8& 7.0 & 26.9 & 103.6 & 432.1 & 1662.8 & 6479.2\\ \hline
\end{tabular}
\end{table*}

\begin{table*}[t]
\scriptsize
\centering
\setlength{\abovecaptionskip}{0pt}
\setlength{\belowcaptionskip}{4pt}
\caption{The computational cost (in seconds) of CG-OMPR under different precisions.}
\label{Table: OMPR precision}
\doublerulesep=2pt
\begin{tabular}[tc]{|c||c|c|c|c|c|c|c|c|c|c|c|c|c|c|}
\hline
{\LARGE \textcolor{white}{o}}N& $2^{11}$& $2^{12}$& $2^{13}$ & $2^{14}$& $2^{15}$ & $2^{16}$ & $2^{17}$ & $2^{18}$ & $2^{19}$ & $2^{20}$ & $2^{21}$ & $2^{22}$ & $2^{23}$ & $2^{24}$  \\ \hline\hline
Exact precision & 0.2&  0.3 & 0.6 & 1.0 & 1.9 & 4.1 & 10.5 & 23.6 & 51.8 & 111.3 & 232.1 & 495.3 & 1109.1 & 2392.3\\ \hline
$\| r_{j} \|_{2} \leq 10^{-10}$ & 0.1& 0.2& 0.3 & 0.5 & 1.0 & 2.2 & 4.8 & 10.1 & 23.6 & 49.8 & 113.9 & 249.3 & 532.3 & 1130.4 \\ \hline
$\| r_{j} \|_{2} \leq 10^{-5}$ & 0.06& 0.09& 0.2 & 0.3 & 0.6& 1.3 & 2.9 & 6.0 & 13.1 & 27.7 & 60.3 & 131.3 & 293.4 & 643.6 \\ \hline
\end{tabular}
\end{table*}

\begin{table*}[t]
\scriptsize
\centering
\setlength{\abovecaptionskip}{0pt}
\setlength{\belowcaptionskip}{4pt}
\caption{The computational cost (in seconds) of CG-SP under different precisions.}
\label{Table: SP precision}
\doublerulesep=2pt
\begin{tabular}[tc]{|c||c|c|c|c|c|c|c|c|c|c|c|c|c|c|}
\hline
{\LARGE \textcolor{white}{o}}N& $2^{11}$& $2^{12}$& $2^{13}$ & $2^{14}$& $2^{15}$ & $2^{16}$ & $2^{17}$ & $2^{18}$ & $2^{19}$ & $2^{20}$ & $2^{21}$ & $2^{22}$ & $2^{23}$ & $2^{24}$  \\ \hline\hline
Exact precision & 0.3&  0.5 & 0.9 & 1.5 & 3.0 & 6.2 & 16.6 & 37.9 & 79.4 &  173.2 & 359.2 & 740.6 & 1581.3 & 3195.4\\ \hline
$\| r_{j} \|_{2} \leq 10^{-10}$ & 0.2 & 0.3& 0.5 & 0.8 & 1.6 & 3.6 & 7.8 & 17.3 & 36.1 & 78.3 & 177.6 & 382.9 & 784.5 & 1706.3\\ \hline
$\| r_{j} \|_{2} \leq 10^{-5}$ & 0.09& 0.2& 0.4 & 0.6 & 0.9 & 2.1 & 5.2 & 10.0 & 22.5 &49.1 &109.2 & 241.6 & 513.2 & 1096.8\\ \hline
\end{tabular}
\end{table*}

\section{Conclusions}\label{sec:conl}
The bottleneck of greedy algorithms in the context of compressive sensing is to solve the least square problem, in particular, when facing large-scale data.
In this paper, we address this challenging issue and propose a fast but cost-effective least square solver.
Our solution has been theoretically proved and can be readily incorporated with the existing greedy algorithms to improve their performance by significantl reducing computation complexit and memory cost.
Case studies on combining our method and OMP, SP, and OMPR have been conducted and shown promising results.

\section{Acknowledgment}
This work was supported by Ministry of Science and Technology, Taiwan, ROC, under grants MOST 104-2221-E-001-019-MY3 and NSC 104-2221-E-001-030-MY3.

\bibliographystyle{IEEEbib}	
\bibliography{refs}		

\end{document}